 \documentclass[11pt, reqno]{amsart}

\usepackage[utf8x]{inputenc}

\usepackage{amsmath, amsthm, amsfonts, amssymb}
\usepackage{epsfig, enumerate}
\usepackage{color}
\usepackage{dsfont}

\usepackage{pdf14}

\newtheorem{Theo}{Theorem}

\newtheorem{proposition}{Proposition}
\newtheorem{lemma}[proposition]{Lemma}

\newtheorem{definition}{Definition}

\newtheorem*{rem*}{Remark}

\newcommand{\I}{{\mathbf I}}

\newcommand{\bP}{{\mathbf P}}

\newcommand{\Z}{{\mathbf Z}}

\newcommand{\vv}{{\vec v}}
\newcommand{\vw}{{\vec w}}
\newcommand{\vx}{{\vec x}}

\newcommand{\Aa}{{\mathcal A}}
\newcommand{\Bb}{{\mathcal B}}
\newcommand{\Cc}{{\mathcal C}}

\newcommand{\Ff}{{\mathcal F}}

\newcommand{\Hh}{{\mathcal H}}

\newcommand{\Qq}{{\mathcal Q}}
\newcommand{\Rr}{{\mathcal R}}
\newcommand{\Ss}{{\mathcal S}}

\newcommand{\Uu}{{\mathcal U}}
\newcommand{\Vv}{{\mathcal V}}

\newcommand{\Xx}{{\mathcal X}}

\newcommand{\Ab}{{\mathbb A}}

\newcommand{\CC}{{\mathbb C}}

\newcommand{\EE}{{\mathbb E}}

\newcommand{\GG}{{\mathbb G}}

\newcommand{\KK}{{\mathbb K}}

\newcommand{\NN}{{\mathbb N}}

\newcommand{\RR}{{\mathbb R}}

\newcommand{\TT}{{\mathbb T}}

\newcommand{\ZZ}{{\mathbb Z}}

\newcommand{\Ccc}{{\mathfrak C}}

\newcommand{\Eee}{{\mathfrak E}}

\newcommand{\aaa}{{\mathfrak a}}
\newcommand{\bbb}{{\mathfrak b}}
\newcommand{\ccc}{{\mathfrak c}}
\newcommand{\ddd}{{\mathfrak d}}

\newcommand{\nul}{{\bf 0}}

\newcommand{\qtx}[1]{\quad\text{#1}\quad}

\newcommand{\GL}{{\rm GL}}

\newcommand{\Her}{{\rm Her}}

\newcommand{\pmat}[1]{\begin{pmatrix} #1  \end{pmatrix}}
\newcommand{\smat}[1]{\left( \begin{smallmatrix} #1  \end{smallmatrix} \right)}

\renewcommand{\CC}{\mathds{C}}
\renewcommand{\RR}{\mathds{R}}
\renewcommand{\ZZ}{\mathds{Z}}
\renewcommand{\NN}{\mathds{N}}

\numberwithin{proposition}{section}
\numberwithin{equation}{section}

\usepackage{hyperref}

\title[A.C.  spectrum for operators with random decaying potentials on the strip]{Absolutely continuous spectrum for Schrödinger operators with random decaying matrix potentials on the strip}

\author{Hernán González}
\author{Christian Sadel}
\address{Facultad de Matem\'aitcas, Pontificia Universidad Cat\'olica de Chile} 
\email[González]{higonzal@mat.uc.cl}
\email[Sadel] {chsadel@mat.uc.cl}


\begin{document}

\begin{abstract}
We consider a family of random Schrödinger operators on the discrete strip with decaying random $\ell^2$ matrix potential.
We prove that the spectrum is almost surely pure absolutely continuous,  apart from random possibly embedded eigenvalues, which may accumulate at band edges.
\end{abstract}

\maketitle

\tableofcontents


\section{Model and main result} 


We consider a random family of block-Jacobi operators on $\ell^2(\ZZ_+)\otimes \CC^l$ given by
\begin{equation}\label{eq-def-H}
(H_{\omega} \Psi)_n =-\Psi_{n-1}-\Psi_{n+1}+A\Psi_n+V_n(\omega) \Psi_n
\end{equation}
where   $\Psi=(\Psi_n)_{n\geq 0}\in \ell^2(\ZZ_+)\otimes \CC^l$ means that $\Psi_n \in \CC^l, \,\forall n \in \ZZ_+$, with $\displaystyle \sum_{n\geq 0} ||\Psi_n ||^2 < \infty$,  In the case $n=0$ one sets $\Psi_{-1}=0$ in \eqref{eq-def-H}.
$A$ is a fixed Hermitian $l\times l$  matrix ($A=A^*$) and finally we have  a random Hermitian-matrix potential $V_n=V_n(\omega)$   This means,  we have some probability space $(\Omega,\Aa,\bP)$ and $\Her(l)$ valued random variables $V_n\,:\, \Omega \to \Her(l)$.
Moreover,  we assume that the family $(V_n)_n$ is independent
and that
\begin{equation}\label{eq-cond-V}
 \sum_{n \geq 0} \left( \| \EE(V_n)\| +\EE (||V_n||^2)\;\right)  \, <\, \infty\;,
\end{equation}
where  $\EE$ denotes the expectation value.  

 We also define the 'unperturbed' operator $H_0$ by eliminating the $V_n$,
\begin{equation}
(H_{0} \Psi)_n =-\Psi_{n-1}-\Psi_{n+1}+A\Psi_n\;.
\end{equation}

$H_0$ and $H_\omega$ can be seen as quasi-one dimensional discrete Schrödinger operators on a semi-infinite strip of width $l$.
The matrix $A$ maybe the adjacency matrix of a finite graph $\GG$,  in which case $H_0$ would be like a discrete Laplace operator on the product graph $\ZZ_+ \times \GG$.  $H_\omega$ is then a random perturbation of $H_0$ adding the matrix potentials $V_n$ at each level $n$.  
This way,  $H_\omega$ falls into the class of operators describing randomly perturbed quantum systems.
The study of such systems was initiated by Anderson \cite{Anderson} with the today called Anderson model where one  studies operators on $\ZZ^d$ with independent identically distributed potentials on each lattice site. 
In general for such models one finds Anderson localization at large disorder (large variance of the potential) and at the edges of the spectrum. Anderson localization means one has pure point spectrum and exponentially decaying eigenfunctions.
There are two general methods to prove this, the fractional moment method \cite{AM} and multi-scale analysis \cite{FS, GK1, GK2}.
The fractional moment method at high disorder works fine in any graphs with a finite upper bound on the connectivity of one point \cite{Tau}. 
In $d=1$ dimension (line or strip) one finds localization for the Anderson model at any disorder \cite{GMP, KuS, CKM, KlLS}.

Except in one dimension,  for a long time the high disorder Bernoulli Anderson model could not be handled,  this means the i.i.d.  potential has a Bernoulli distribution.
A first breakthrough was done for the continuous model in \cite{BK},  and recently,  the high-disorder localization has also been shown for the discrete Bernoulli Anderson model in $\ZZ^d$ for $d=2$ and $d=3$ dimensions \cite{Li,LZ}.

From $d\geq 3$ dimensions on,  one expects some absolutely  continuous spectrum at small enough disorder. However, this is still conjectural.
Existence of absolutely continuous spectrum for Anderson models at low disorder has first been proved for infinite dimensional hyperbolic type graphs like regular trees and tree-like structures \cite{Kl, ASW, AW, FHS, FHS2, KS, KLW, Sa-FC, Sa-Fib}.
It has also been shown for the Anderson model on special graphs with a finite-dimensional growth,  so called anti-trees and partial antitrees \cite{Sa-AT, Sadel}.

\vspace{.2cm}

As a mean to study critical transitions from absolutely continuous to pure point spectrum, random decaying potentials in one dimension were also investigated \cite{KiLS,Lasi,FHS3}.  Here,  we extend and improve on the result by Froese,  Hasler and Spitzer \cite{FHS3} using methods similar to Last and Simon \cite{Lasi}.
The key point for the absolutely continuous spectrum result in \cite{Lasi} has been the spectral average formula by Carmona-Lacroix \cite[Theorem II.3.2]{Cala}.  Here, we use its generalization to strips,  Proposition~\ref{th-spec-av-strip} which is a special case of the broader generalization recently done in \cite{Sadel2019}.

\subsection{Spectrum and spectral bands}

Without loss of generality,  we may assume that $A$ in \eqref{eq-def-H} is a diagonal matrix:  If this is not the case, then, as $A^*=A$,   there is a unitary matrix $U$ such that $U^*AU$ is diagonal.  
 Then,  define the unitary operator
 $\Uu:\ell^2(\ZZ_+)\otimes  \CC^l \to \ell^2(\ZZ_+) \otimes  \CC^l $ by
  $(\Uu\Psi)_n:=U\Psi_n $ and one finds:
  $$(\Uu^* H_{\omega} \Uu \Psi)_n =-\Psi_{n+1}-\Psi_{n-1}+U^*AU\Psi_{n}+U^*V_nU\Psi_n$$
 Now $U^*AU$ is diagonal and $U^*V_n U$ are random Hermitian matrices satisfying an inequality as \eqref{eq-cond-V}.
Thus,   using this unitary conjugation, we may assume that $A$ is diagonal,  hence
 \begin{equation} 
 A=\left(\begin{matrix}
\alpha_1 & 0  & \hdots  & 0   \\
0 & \alpha_2  & 0 & 0   \\
\vdots & 0  &\ddots & 0   \\
0 & 0  & 0  &  \alpha_l 
\end{matrix}\right) 
\end{equation}
with $\alpha_j \in \RR$, $j=1,\ldots,l$ being the eigenvalues of $A$.
As a consequence, 
$$
\sigma(H_0)=\bigcup_{j=1}^l [a_j-2,a_j+2]$$ 
and the spectrum  is purely absolutely continuous.

We call $[\alpha_j-2,\alpha_j+2]$ the j-th band of the spectrum of $H_0$,  $\{\alpha_j-2,\alpha_j+2\}$ are the band-edges of this band.  Each band-edge can be internal,  meaning inside of another band, or external, meaning an edge (boundary point) of the spectrum of $H_0$.  We consider the spectrum of $H_0$ without all the (external and internal) band-edges and define
\begin{equation}
\Sigma = \left[\bigcup_{j=1}^l (\alpha_j-2, \alpha_j+2) \right] \setminus \left[\bigcup_{j=1}^l \{\alpha_j-2,\alpha_j+2 \} \right]
\end{equation}

Note $\Sigma$ is open and $\overline{\Sigma}=\sigma(H_0)$.
We also define the intersection of all open bands,
\begin{equation}
\Sigma_0=\bigcap_{j=1}^l (\alpha_j-2, \alpha_j+2)\;
\end{equation}
which might be empty.
For the essential spectrum we note that $H_\omega=H_0+\bigoplus_n V_n$ where $\bigoplus_n V_n$ is almost surely a compact operator, hence,

$$\sigma_{\rm ess}(H_{\omega})=\overline{\Sigma}=\sigma(H_0)$$

\subsection{The main result}

The main theorem of the whole thesis is the following:
\begin{Theo}\label{th:main}
Apart from discrete spectrum,  (embedded isolated eigenvalues)  the spectrum of $H_{\omega}$ is almost surely purely absolutely continuous in $\Sigma$.  Moreover, there are no embedded eigenvalues in the intersection of the bands,  $\Sigma_0$.
That means,  there may be random embedded eigenvalues in $\Sigma \setminus \Sigma_0$ which may only accumulate at the boundary $\partial \Sigma$,  that is, 
the internal and external band-edges.   \\
In technical terms, this means, 
there is a set $\hat \Omega \subset \Omega$ of probability one,  $\bP(\hat \Omega)=1$,  such that for all $\omega \in \hat \Omega$ and all compact subsets
$\Ccc \subset \Sigma$,  there is a finite (random) subset of eigenvalues $\Eee=\Eee(\omega) \subset \Ccc \setminus \Sigma_0$,  such that the spectrum of $H_\omega$ is purely absolutely continuous in $\Ccc \setminus \Eee$.
\end{Theo}
Under the slightly stronger assumption that $\EE(V_n)=0$ and $\sum_n \EE(\|V_n\|^2+\|V_n\|^4)<\infty$ it was aready shown in \cite{FHS3} 
that the spectrum is purely absolutely continuous in $\Sigma_0$, and that there is absolutely continuous spectrum in all of $\Sigma$.
However,  the proof method used there for the set $\Sigma$ does not exclude any other type of singular spectrum.

Note that in the line case, $l=1$ we have $\Sigma=\Sigma_0$ and purely absolutely continuous spectrum in this case has already been shown in
\cite{Lasi}. On the line case it is also known that for any,  also non-random $\ell^2$-potential,  one has absolutely continuous spectrum in $\Sigma$ \cite{DK}, but again, any other type of embedded singular spectrum is possible (not excluded in the proof).

The general operator $H_\omega$ investigated here allows the case,  were the operator (almost surely) splits into the direct sum of two strip operators $H^1 \oplus H^2$ (two separated strips). Then, adjusting one of the $V_n$ one may create an eigenvalue for $H^1$, lying outside of its essential spectrum, but lying inside the essential spectrum of $H^2$. In fact, one may have the part of $V_n$ belonging to $H^1$ non-random and create some fixed embedded eigenvalue (for all of $\omega$).
Thus, without further 'channel-mixing' assumptions, one can not expect to obtain pure absolutely continuous spectrum within $\Sigma$.
But under sufficient 'mixing' created by the $V_n$, this should be true.

Finitely, let us mention that Theorem~\ref{th:main} does also cover the  'full' strip case going from $-\infty$ to $\infty$.
That means, Theorem~\ref{th:main} also applies to random operators $\Hh_\omega$ on $\ell^2(\ZZ)\otimes \CC^l$ (rather than $\ell^2(\ZZ_+)$) of the form
$$
(\Hh_\omega \Psi)_n=-\Psi_{n+1}-\Psi_{n-1}+\Aa \Psi_n + \Vv_n \Psi_n
$$
where the $\Vv_n$ are independent Hermitian matrices satisfying
$$
\sum_{n=-\infty}^\infty \left( \| \EE(\Vv_n) \| + \EE(\|\Vv_n\|^2) \right) <\infty\;.
$$
To see this, we can transform the operator $\Hh$ unitarily to an operator $H_\omega$ on $\ell^2(\ZZ_+)\otimes \CC^{2l}$ of the above form \eqref{eq-def-H} by defining
$$
A=\pmat{\Aa & 0 \\ 0 & \Aa}, \quad V_n=\pmat{\Vv_n & 0 \\ 0 & \Vv_{-n-1}} \quad \text{for $n \geq 1$,\; and}\quad V_0=\pmat{\Vv_0 & -I \\ -I & \Vv_{-1}}
$$
It is obvious to see that the random matrices $V_n$ do indeed satisfy the conditions as above.

To picture the transformation,  just think of the case $l=1$, the doubly infinite discrete line,  and then flip over the negative line to make a half-infinity strip of width two. The off-diagonal blocks in $V_0$ then correspond to the connection in the $n=0$ shell of the strip, which corresponds to the points $0$ and $-1$ on the line:

\begin{center}
\includegraphics[width=8cm]{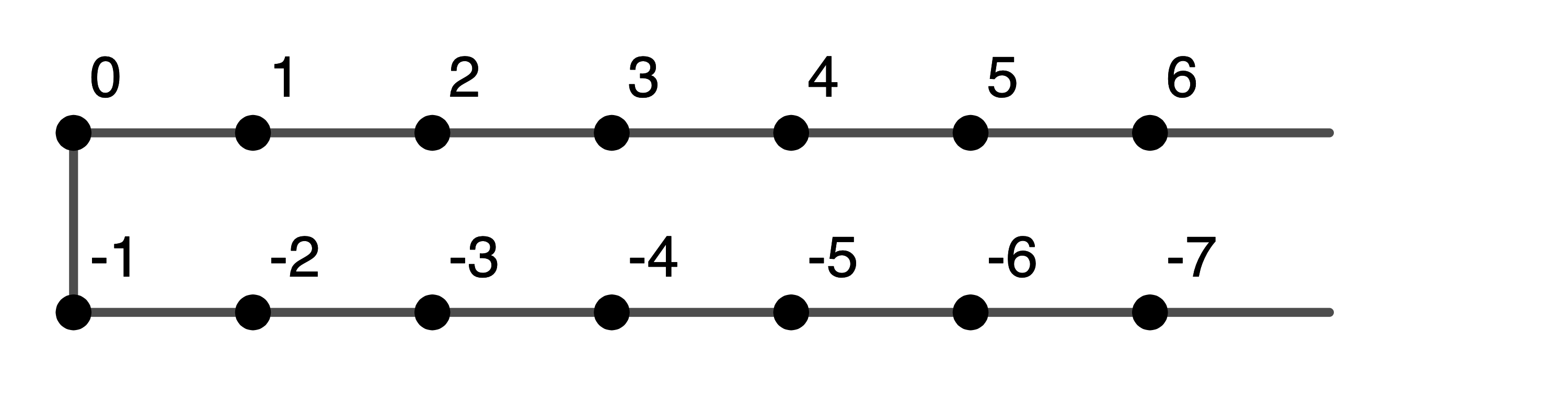}
\end{center}

\section{Transfer matrices, elliptic and hyperbolic channels \label{sec:3-channels}}

The eigenvalue equation $ H_{\omega}\Psi=z\Psi$ is a recursion that can be written in the matrix form as follows:
$$\left(\begin{matrix}
\Psi_{n+1}\\
\Psi_{n}
\end{matrix}\right)=T^z_{n} \left(\begin{matrix}
\Psi_{n}\\
\Psi_{n-1}
\end{matrix}\right) \qtx{where}
T^z_n=\begin{pmatrix} V_n+A-z I & - I \\ I & 0\end{pmatrix}\;.
$$
Iteration leads to the products of transfer matrices for $n>m$,
$$
\pmat{\Psi_{n+1} \\ \Psi_n }=T^z_{m,n} \pmat{\Psi_m \\ \Psi_{m-1}} \qtx{where} T^z_{m,n}=T^z_n T^z_{n-1} \cdots T^z_{m+1} T^z_m\;.
$$
We may write
$$
T^\lambda_n=T_{H_0}^\lambda+\pmat{V_n & 0 \\ 0 & 0} \qtx{where}
T^\lambda_{H_0}=\pmat{A-\lambda I & -I \\ I & 0}
$$
is basically the transfer matrix of the unperturbed operator $H_0$.
We will now write the transfer matrix in some basis which diagonalises $T_{H_0}^\lambda$.
Recall,  $A$ is assumed diagonal and its eigenvalues are $\alpha_1, \ldots, \alpha_l$.  Adopting the notions of \cite{Sadel2011} we define:
\begin{definition} Let $\lambda \in \RR$.  We call the $j$-th channel 
\begin{enumerate}
 \item \textbf{Elliptic} at $\lambda$ if $|\alpha_j-\lambda|<2$
\item \textbf{Hyperbolic} at $\lambda$ if $|\alpha_j-\lambda|>2$
\item \textbf{Parabolic} at $\lambda$ if $|\alpha_j-\lambda|=2$ 
\end{enumerate}
\end{definition}

Now fix some $\lambda \in \Sigma$.  Note that by the definition of $\Sigma$,  there are no parabolic channels and there is at least one elliptic channel at $\lambda$ 
We assume the channels to be ordered such that:
$$|\alpha_j-\lambda|<2 \quad  \forall j \in \{1,...,l_e\} $$
$$|\alpha_j-\lambda|>2 \quad \forall j \in \{l_e+1,...,l\} $$
Note that the set of all $\lambda$ satisfying these inequalities  is some open interval $(\lambda_0, \lambda_1) \subset \Sigma$. We later vary $\lambda$ slightly within this interval.  
For $\lambda\in (\lambda_0, \lambda_1)$,  and $ j\in \{1,...,l_e\} $  we define $k_j = k_j(\lambda) \in (0,\pi)$ by 
$$2\cos(k_j)=\alpha_j-\lambda\;.$$ 
For $j\in \{1,...,l_h\} $ with $l_h=l-l_e$ we define  $\gamma_j = \gamma_j(\lambda) \in \RR$,  $|\gamma_j|>1$,  by
$$\gamma_j + \frac{1}{\gamma_j }=\alpha_{j+l_e}-\lambda $$
We define the diagonal matrices $$\Gamma= \left(\begin{matrix}
\gamma_1 & 0  & \hdots  & 0   \\
0 & \gamma_2  & 0 & 0   \\
\vdots & 0  &\ddots & 0   \\
0 & 0  & 0  &  \gamma_{l_h}  
\end{matrix}\right)\,,\quad K= \left(\begin{matrix}
k_1 & 0  & \hdots  & 0   \\
0 & k_2  & 0 & 0   \\
\vdots & 0  &\ddots & 0   \\
0 & 0  & 0  &  k_{l_e}  
\end{matrix}\right) $$ 
such that
 $$A-\lambda I =\left(\begin{matrix}
2\cos(K) & 0  \\
0 & \Gamma + \Gamma^{-1}
\end{matrix}\right)\;. $$
Thus,  for the transfer matrices of $H_0$ we find
$$T^{\lambda}_{H_0} =  \left(\begin{matrix}
2\cos(K)  & 0  &-I  & 0   \\
0 &   \Gamma+\Gamma^{-1}  & 0 & -I  \\
I & 0  &0 & 0   \\
0 & I  & 0  & 0 
\end{matrix}\right)\;.$$

Also we note that $$ T^\lambda_{H_0} \left(\begin{matrix}
e^{\pm i K}  \\ 0  \\I  \\ 0  \end{matrix}\right) = \left(\begin{matrix}
e^{\pm i K}  \\ 0  \\I  \\ 0  \end{matrix}\right) e^{\pm i K}, \qquad
T^\lambda_{H_0} \left(\begin{matrix}
0 \\ \Gamma^{\pm 1} \\ 0  \\ I  \end{matrix}\right) =\left(\begin{matrix}
0 \\ \Gamma^{\pm 1} \\ 0  \\ I  \end{matrix}\right) \Gamma^{\pm 1} $$
so that $e^{\pm ik_j}$, $\gamma_j$ and $\gamma_j^{-1}$ are the eigenvalues of $T^\lambda_{H_0}$.
In order to diagonlize $T^\lambda_{H_0}$ we introduce  
$$
Q_\lambda = \pmat{e^{iK} &  e^{-iK} & 0 & 0 \\ 0 & 0 & \Gamma^{-1} & \Gamma \\ I_{l_e} & I_{l_e} & 0 & 0 \\ 0 & 0 & I_{l_h} & I_{l_h} }\;,
$$
where $I_d$ is the unit matrix of size $d \times d$,  then
\begin{equation}\label{eq-QTQ}
Q_\lambda^{-1} T^\lambda_n Q_\lambda\,=\,  T^\lambda \,+\,
\Vv^\lambda_n
\end{equation}
with $T^\lambda$ being diagonal, more precisely,   
\begin{equation}\label{eq-T^l}
T^\lambda\,=\,\pmat{e^{iK} & 0 & 0 & 0 \\ 0 & e^{-iK} & 0 & 0 \\ 
0 & 0 & \Gamma^{-1} & 0 \\ 0 & 0 & 0 & \Gamma}\;,\qquad 
\Vv^\lambda_n = Q_\lambda^{-1} \pmat{ V_n & 0 \\ 0 & 0} Q_\lambda\;.
\end{equation}

We note that $Q_\lambda$ is indeed invertible for $\lambda \in (\lambda_0, \lambda_1)$ as
$e^{i k_j} \neq e^{-ik_j}$ and $\gamma_j \neq 1/\gamma_j$ in this case.
Defining
\begin{equation}
\Qq_K = (e^{iK}-e^{-iK})^{-1}\;, \quad  \Qq_\Gamma = (\Gamma^{-1} - \Gamma)^{-1}
\end{equation}
we find
\begin{equation}\label{eq-Q^-1}
Q_\lambda^{-1}\,=\,
\pmat{ \Qq_K & 0 & -e^{-iK} \Qq_K & 0 \\ -\Qq_K & 0 & \Qq_K e^{iK} & 0 \\ 0 & \Qq_\Gamma  & 0 & - \Gamma  \Qq_\Gamma \\ 0 & - \Qq_\Gamma & 0 & \Gamma^{-1}  \Qq_\Gamma }
\end{equation}

Now,  in order to work with uniform estimates we will restrict our consideration to a compact interval $[a,b] \subset (\lambda_0, \lambda_1) \subset \Sigma$.  Chosen such a compact interval and allowing complex values for $k_j$ and $\gamma_j$,  we can extend the definitions of 
$K=K(\lambda), \Gamma=\Gamma(\lambda), Q_\lambda,  Q_\lambda^{-1}, T^\lambda$ analytically to spectral parameters $z$ in the complex plane, $z=\lambda+i\eta \in [a,b]+i[-c,c]\subset \CC$,  for $c$ small enough.  This means $\lambda \in [a,b], \; \eta \in[-c,c]$, 
 and the equations \eqref{eq-QTQ} and \eqref{eq-T^l} still hold with $\lambda$ replaced by $z$. 
  We will need this extension in some part to use analyticity arguments.

Choosing $c>0$ small enough,  one can guarantee by compactness and analyticity arguments,  that there is some $\gamma>0$ such that 
\begin{equation} \label{eq-Gamma-estimate}
\left\| \Gamma^{-1}(\lambda+i\eta) \right\|\,\leq\,  e^{-2\gamma}\;, \qtx{and}
\left\| e^{\pm iK(\lambda+i\eta)} \right\|\,\leq\, e^{\gamma}\quad
\forall\, \lambda \in [a,b], \; \forall \,|\eta|\leq c\;.
\end{equation}
Note for $\lambda \in [a,b]$ we have $\|e^{\pm i K}\|=1$.

\section{The key estimates \label{sec:3-key-est}}

The estimates in this section are somewhat independent of the rest of the paper.
But in many respects, it is the key part of the proof.

\vspace{.5cm}

We consider the following general situation:
Let be given independent random $(l_0+l_1) \times (l_0+l_1)$ matrices of the form
\begin{equation}\label{eq-setup-T}
T_n\,=\,T\,+\,W_n \qtx{where} T=\begin{pmatrix}
S &   \\
 & \Gamma  \\
\end{pmatrix}\;,\quad S\in \CC^{l_0 \times l_0} , \quad \Gamma \in \CC^{l_1 \times l_1}
\end{equation}
where for some fixed $\gamma>0$ we have
\begin{equation}\label{eq:cond1}
    \|S\|\,\leq\,e^\gamma \qtx{,} \|\Gamma^{-1}\| \leq e^{-2\gamma} \;.
\end{equation}
(Later we use $l_1=l_h$ and $l_0=2l_e+l_h$.)

Note that the second condition implies $\|\Gamma v \|\geq e^{2\gamma} \|v\|$ for any vector $v\in\CC^{l_0+l_1}$.
Moreover, $W_n$ are independent random $(l_0+l_1) \times (l_0+l_1)$ matrices satisfying, with some fixed constant $\Cc_W>0$, 
\begin{equation} \label{eq:cond2}
    \|W_n\| \leq \frac{e^{2\gamma}-e^{\gamma}}{4} \qtx{and} \sum_{n=1}^\infty \| \EE(W_n) \| + \EE(\|W_n\|^2)\,\leq \Cc_W \,<\,\infty
\end{equation}
For certain parts we will also assume the stricter bound $\|S\| \leq 1$ (cf. Proposition~\ref{prop-limits}). In fact, at real spectral parameters $\lambda$ we will have this bound, however, for some arguments we need to allow some small imaginary part, which is why in general we only assume $\|S\| \leq e^\gamma$ in this section.

Now let us consider the Markov process of $(l_0+l_1)\times (l_0+l_1)$ matrices given by
$$
\Xx_0=I,\quad \Xx_{n+1} = T_n\, \Xx_{n}\;.
$$
Using the splitting into blocks of sizes $l_0$ and $l_1$ like above we write
\begin{equation}
    \Xx_n=\begin{pmatrix} A_n & B_n \\ C_n & D_n
    \end{pmatrix} \qtx{and} W_n=\begin{pmatrix} a_n & b_n \\ c_n & d_n \end{pmatrix}
\end{equation}
From the process $\Xx_n$ we will define the process of pairs of matrices $(X_n, Z_n)$ given by
\begin{equation}
    X_n=A_n-B_nD_n^{-1} C_n\;,\quad Z_n=B_n D_n^{-1}\;.
\end{equation}
$X_n$ is a so called Schur-complement. Some standard calculations , see for instance \cite{SadelV},  show that $(X_n, Z_n)$ can be seen as the process of equivalence classes of $\Xx_n$ defining
$$
\Xx_1 \sim \Xx_2 \quad \Leftrightarrow \quad \Xx_1 = \Xx_2 \begin{pmatrix} I & 0 \\ M & G\end{pmatrix}
$$
with $I$ being the $l_0 \times l_0$ identity matrix, $G\in \GL(l_1)$ and $M$ being any $l_1 \times l_0$ matrix. Note that the set of matrices of the form $ \left(\begin{smallmatrix} I & 0 \\ M & G\end{smallmatrix}\right)$ is a group.

In that sense if $D_n$ is invertible we get
$$
\Xx_{n+1}=\pmat{A_n & B_n \\ C_n & D_n} \sim 
\pmat{A_n & B_n \\ C_n & D_n} \pmat{I & \nul \\ -D_n^{-1} C_n & D_n^{-1}}=
\pmat{X_n & Z_n \\ \nul & I}\;
$$
and we find
$$
\Xx_{n+1} \sim
\pmat{S+a_n & b_n \\ c_n & \Gamma+d_n} \pmat{X_n & Z_n \\ \nul & I} = 
\pmat{(S+a_n)X_n & (S+a_n)Z_n +b_n \\ c_n X_n & c_n Z_n+\Gamma+d_n}
$$
which leads to the identities
\begin{equation}\label{eq-Z_n+1}
Z_{n+1}=\big((S+a_n)Z_n+b_n\big) \big( c_nZ_n+\Gamma+d_n\big)^{-1}\
\end{equation}
\begin{equation}\label{eq-X_n+1}
X_{n+1}=(S+a_n)X_n-Z_{n+1}\,c_n X_n
\end{equation}
provided that $D_n^{-1}$ and $D_{n+1}^{-1}$ exist.

\vbox{
\begin{proposition}\label{prop-boundZ}
Under the given assumptions \eqref{eq:cond1} and \eqref{eq:cond2} one finds:
$D_n$ is invertible for all $n \in \NN$, hence, $X_n$ and $Z_n$ are well defined for all $n$ and
$$
\sup_{n \in \NN} \|Z_n\| \leq 1$$
\end{proposition} 
}

\begin{proof}
First we note that by straight forward calculations one finds for $c>0$ and some square $d\times d$ matrix $M$ that
\begin{equation}\label{lem:est-inv}
\|Mv\| \geq c \|v\| \qtx{for all $v \in \CC^d$} \Leftrightarrow \quad M\;\;\text{is invertible and}\;\; \|M^{-1}\| \leq \frac1c
\end{equation}
From there we find:
\begin{lemma} \label{lem:estZ-1} 
Assume $\lVert Z_n \rVert \leq 1$ and that the bounds  \eqref{eq:cond1},  \eqref{eq:cond2}  hold.  Then
$$\lVert (c_n Z_n + \Gamma + d_n)^{-1} \rVert\,\leq\, \frac{1}{e^{2\gamma}-2\|W_n\|}\,\leq\, \frac{2}{e^{2\gamma}+e^\gamma}$$
\end{lemma}

To show this,  we use the equivalence \eqref{lem:est-inv} noting
\begin{align*}&\lVert (c_n Z_n + (\Gamma + d_n))v \rVert \geq  \lVert \Gamma v \rVert  - \lVert  (c_n Z_n +d_n)v \rVert \geq e^{2\gamma}\lVert  v \rVert -   (\lVert  c_n\rVert\lVert  Z_n \rVert +\lVert d_n \rVert)\lVert v  \rVert \\ &\qquad  \geq \left( e^{2\gamma}-2 \|W_n\|\right)\|v\|\,\geq\, \left(e^{2\gamma}-\frac{e^{2\gamma}-e^\gamma}2 \right) \lVert v  \rVert
= \frac{e^{2\gamma}+e^\gamma}2\;\|v\|\;.
\end{align*}
Now we proof by induction that $D_n$ is invertible and $\|Z_n\|\leq 1$.
First, we notice $D_0=I$ is invertible and $\|Z_0\|=\|0\|=0 \leq 1$. Now assume $\|Z_n\|\leq 1$ and $D_n$ being invertible. We find
$$ \left(\begin{matrix}
A_{n+1} & B_{n+1}   \\
C_{n+1} & D_{n+1}  
\end{matrix}\right)  = \left(\begin{matrix}
S+a_{n} & b_{n}   \\
c_{n} & \Gamma+d_{n}  
\end{matrix}\right) \left(\begin{matrix}
A_n & B_n   \\
C_n & D_n  
\end{matrix}\right) $$
and by the lower right block
\begin{align*}
     D_{n+1}D_n^{-1}  &= c_n B_n D_n^{-1} + (\Gamma + d_n) \,=\,  c_n Z_n + (\Gamma + d_n)
\end{align*}
Lemma~\ref{lem:estZ-1} now shows invertibility of $D_{n+1} D_n^{-1}$ and hence of $D_{n+1}$.
Using \eqref{eq-Z_n+1} we find
\begin{align*}
    \lVert Z_{n+1}\rVert&=\lVert ((S+a_n)Z_n+b_n) (c_n Z_n + (\Gamma + d_n))^{-1} \lVert \;\leq\; \left(e^\gamma+\frac{e^{2\gamma}-e^\gamma}{2}\right) \,\frac{2}{e^{2\gamma}+e^\gamma}\,=\,1\;.
\end{align*}
This finishes the induction.
\end{proof}

\begin{proposition} \label{prop-limits}
Under the given assumptions \eqref{eq:cond1} and \eqref{eq:cond2} and the aditional property $\lim_{n\to\infty} W_n=0$ we find
$$
\lim_{n\to \infty} Z_n = 0\,,\quad \lim_{n\to \infty} D_n^{-1} = 0 \qtx{and}
\lim_{n \to \infty} D_n^{-1} C_n\,=\,Y \quad \text{exists}\;.
$$
\end{proposition}
\begin{proof} We will prove by induction that there exists $N_k$ such that for all $n>N_k$ we have
$\|Z_n\|\leq e^{-k\gamma/2}$ for all $n>N_k$. The induction start for $k=0$ is given by Proposition~\ref{prop-boundZ}.
Assume the statement is true for $k$. Then we find $N$ such that for all $n>N$
$$\|W_n\|\,<\,e^{-k\gamma/2}\;\frac{e^{3\gamma/2}-e^{\gamma}}{4} \qtx{and} \|Z_n\|\leq e^{-k\gamma/2}\;.
$$
Using \eqref{eq-Z_n+1},  Lemma~\ref{lem:est-inv} we find for $n>N$ that
\begin{align*}
\|Z_{n+1}\|\,&\leq\,\;\frac{ (e^\gamma+\|W_n\|)\|Z_n\|+\|W_n\|}{e^{2\gamma}-2\|W_n\|}\,\leq\
\\[.2cm]
&\leq\,\,e^{-k\gamma/2}\;
\frac{e^\gamma+\frac14(e^{3\gamma/2}-e^\gamma)(1+e^{-k\gamma/2})}{e^{2\gamma}-\frac12(e^{3/2\gamma}-e^\gamma)e^{-k\gamma/2}}\,\\[.2cm]
&\leq\,e^{-k\gamma/2}\;\frac{e^\gamma+\frac12(e^{3\gamma/2}-e^\gamma)}{e^{2\gamma}-\frac12(e^{2\gamma}-e^{3\gamma/2})}\,=\, e^{-(k+1)\gamma/2}
\end{align*}
For the last line we used the estimates $e^{-k\gamma/2} \leq 1 \leq e^{\gamma/2}$. 
This finishes the induction and the first statement.

For the second statement, note
$$
\|D_{n+1}^{-1}\| \leq \|D_n^{-1}\|\,\|[\Gamma+d_n+c_n Z]^{-1}\|\,\leq\, \frac{2}{e^{2\gamma}+e^\gamma} \|D_n\|^{-1}\;.
$$
which gives 
\begin{equation}
\|D_n^{-1}\|\,\leq\,  \left( \frac{e^{2\gamma}+e^\gamma}{2}\right)^{-n}\,\to\, 0
\end{equation}
where we use that $D_0=I$ as $\Xx_0=I$.  

Moreover,to get the last assertion,  note
\begin{align*}
D_{n+1}^{-1} C_{n+1}\,&=\,[(\Gamma+d_n)D_n+c_nB_n]^{-1}[c_nA_n+(\Gamma+d_n)C_n]\,=\,\\
&=\, D_n^{-1} \left[ \Gamma+d_n+c_n Z_n\right]^{-1}[ c_nA_n+(\Gamma+d_n)C_n]
\end{align*}
Now, using
\begin{align*}
&\left[ \Gamma+d_n+c_n Z_n\right]^{-1}(\Gamma+d_n)
=\, I-\left[\Gamma+d_n+c_nZ_n \right]^{-1} c_n Z_n
\end{align*}
we obtain
\begin{align*}
D_{n+1}^{-1} C_{n+1}\,&=\, D_n^{-1} C_n + D_n^{-1} \left[ \Gamma+d_n+c_n Z_n\right]^{-1} c_n (A_n-Z_n C_n)\\
&=\, D_n^{-1} C_n\,+\,D_{n+1}^{-1} c_n X_n\;.
\end{align*}
Therefore, using $D_0=I, C_0=0$,
\begin{equation}
D_{n+1}^{-1}C_{n+1}=\sum_{k=0}^n D_{k+1}^{-1} c_k X_k\;.
\end{equation}

Using \eqref{eq-X_n+1} and $\|Z_n\|\leq 1$ and the condition \eqref{eq:cond1} we obtain
$$
\|X_{n+1}\|\,\leq\, (e^\gamma+2\|W_n\|) \|X_n\| \;.
$$
We note that for $\varepsilon>0$ sufficiently small we have $e^\gamma+2\varepsilon < \frac{e^{2\gamma}+e^\gamma}2$.
Now there exists $N>0$ such that for $n>N$ we have $\|W_n\|<\varepsilon$ and thus for $n>N$ we find for $\Cc_0=\|X_N\|$ that
$$
\|D_{n+1}^{-1} c_n X_n\| \leq \Cc_0 {\underbrace{\left(\frac{2(e^\gamma+2\varepsilon)}{e^{2\gamma}+e^\gamma}\right)}_{<1}}^n \varepsilon\;.
$$
Therefore,
$$
\sum_{n=0}^\infty D_{n+1}^{-1} c_n X_n
$$
converges absolutely and
$$
\lim_{n\to \infty} D_n^{-1} C_n = \sum_{k=0}^\infty D_{k+1}^{-1} c_k X_k 
$$
exists.
\end{proof}

Concerning probabilistic estimates,  the main point of this section is the following proposition.

\begin{proposition}\label{prop-boundX}
Under the given assumptions \eqref{eq:cond1} and \eqref{eq:cond2} and the additional condition $\|S\|\leq 1$,  one has
$$
\sup_n \, \EE(\|X_n\|^4)\,\leq \Cc_{W,\gamma}\,<\,\infty
$$
where $\Cc_{W,\gamma}$ is a continuous function in $\gamma>0,$ and $\Cc_W>0$ as they appear in \eqref{eq:cond1}, \eqref{eq:cond2}.
\end{proposition} 
\begin{proof}
Given a starting vector $v_0 \in \CC^{l_0}$ we define $(v_n)_n$ inductively by
$v_{n}\,=\,X_{n}\,v_0\;.$
Using \eqref{eq-X_n+1}  we find
\begin{align*}
    \lVert v_{n+1} \rVert^2 &=\langle v_0^{\ast} X_{n+1}^{\ast}, X_{n+1} v_0 \rangle\\
    &= v_0^* X_n^*[ S^*+a_{n}^*-c^*_{n}Z^*_{n+1}]  [S+a_{n}-Z_{n+1}c_{n}]X_n v_0 \\
    &= \underbrace{v_n^{\ast}S^* S v_n}_{\chi_1}+\underbrace{2\Re e (v_n^{\ast} S^* a_n  v_n)}_{\chi_2}+\underbrace{-2\Re e [v_n^{\ast} S^* Z_{n+1} c_n v_n]}_{\chi_3}\\
    &\;\;\;\;+\underbrace{v_n^{\ast} (a_n^*- c_n^* Z_{n+1}^*)(a_n-Z_{n+1}c_n) v_n}_{\chi_4}
\end{align*}

Now: 
\begin{align*}
    \lVert v_{n+1} \rVert^4 &= \chi_1^2+ \chi_2^2+\chi_3^2 +\chi_4^2 + 2\sum_{j=2}^{4} \sum_{i=1}^{j-1}\chi_i \chi_j 
\end{align*}

As $||S||\leq 1$ and $\|Z_n\|\leq 1$, we first note
\begin{align*}
 |\chi_1|\,&\leq\,  \|v_n\|^2 \\
 |\chi_2|\,&\leq\, 2 \|W_n\| \|v_n\|^2 \\
 |\chi_3|\,&\leq\, \|W_n\| \|v_n\|^2 \\
 |\chi_4|\,& \leq\, 4 \|W_n\|^2\,\|v_n\|^2
\end{align*}
The problematic terms, were we can not use the expectation outside the norm are $\chi_1 \chi_2$ and $\chi_1 \chi_3$.
For the other terms,  we remark
\begin{align}
&\EE(\chi_1^2+\chi_2^2+\chi_3^2+\chi_4^2+2\chi_1 \chi_4+2\chi_2 \chi_3 +2 \chi_2 \chi_4+2 \chi_3 \chi_4)\,\leq\, \nonumber \\ & \qquad\qquad \leq\; \EE\left( \|v_n\|^4 \left(1 +17\|W_n\|^2+24\|W_n\|^3+16\|W_n\|^4\right) \right) \nonumber \\
&\qquad \qquad \leq \; \EE(\|v_n\|^4) \left(1+\EE( \|W_n\|^2)[17+6 (e^{2\gamma}-e^\gamma)+(e^{2\gamma}-e^\gamma)^2)]\right)
\end{align}
For the last step we use the bound \eqref{eq:cond2} and the fact that $W_n$ is independent of $X_n$ and thus $v_n$.

For the frst problematic term, note that
$$
\EE(\chi_1 \chi_2) \,=\,\EE(\EE(\chi_1\chi_2|X_n))=\EE\big(\chi_1 2 \Re e(v_n^* S^* \EE(a_n) v_n)\big)
$$
using the fact that $v_n $ is $X_n$ measurable and $a_n$ is independent of $X_n$.
Thus
\begin{equation}
|\EE(\chi_1 \chi_2)|\,\leq\, 2\, \EE(\|v_n\|^4)\,\| \EE(W_n)\|\;.
\end{equation}
 Now for the term $\chi_1 \chi_3$ we want to use a similar estiamte. However, one of hte problem is now that $Z_{n+1}$ actually depends on $W_n$ and $X_n$.  However,  $W_n$ is independent of $\Xx_n$ and $(X_n, Z_n)$ are $\Xx_n$ measurable. 
 Thus,  we want to condition on $\Xx_n$. Furthermore,  before that, in order to handle some the inverse,  we use we use a resolvent identity together with \eqref{eq-Z_n+1} to find
 $$
Z_{n+1}\,=\,((S+a_n) Z_n+b_n)\left(\Gamma^{-1}- (\Gamma + d_n+c_n Z_n )^{-1}(c_n Z_n+d_n)\Gamma^{-1}\right) \\ $$
giving
$$
Z_{n+1}=((S+a_n )Z_n+b_n)\Gamma^{-1}- Z_{n+1}(c_n Z_n+d_n)\Gamma^{-1}  \;.
$$
Thus,
$$
Z_{n+1}=SZ_n \Gamma^{-1}\,+\,M_n 
$$
where
$$
\|M_n\|\,\leq\, 4\,\|W_n\|\,\|\Gamma^{-1}\| \,\leq\, 4\,e^{-2\gamma}\,\|W_n\| \;.
$$
Splitting up $Z_{n+1}$ this way gives
$$
\EE(\chi_1 \chi_3)\,=\,-2\Re e\,\EE(\chi_1 v_n^* S^* S Z_n \Gamma^{-1} c_n v_n)
\,-\,2 \Re e \EE(\chi_1 v_n^* S^*  M_n c_n v_n)\;.
$$ 
Using the bounds from above, we see
$$
 |2 \Re e \EE(\chi_1 v_n^* S^*  M_n c_n v_n)|\, \leq\,  8 e^{-2\gamma} \EE(\|W_n\|^2 \|v_n\|^4) = 8e^{-2\gamma}\, \EE(\|W_n\|^2) \EE(\|v_n\|^4)\;.
 $$
and
\begin{align*}
&\left|\EE(\chi_1 v_n^* S^* S Z_n \Gamma^{-1} c_n v_n)\right|=
\left| \EE\big(\EE(\chi_1v_n^* S^* S Z_n \Gamma^{-1} c_n v_n | \Xx_n) \big)\right| \\
&\qquad =
\left|\EE(\chi_1 v_n^* S^* S Z_n \Gamma^{-1} \EE(c_n) v_n)\right|\,\leq\,  e^{-2\gamma}\,\| \EE(W_n)\|\,\EE(\|v_n\|^4)
\end{align*}
Thus,  we have in total the bound
\begin{equation}
|\EE(\chi_1 \chi_3)|\,\leq\, \EE(\|v_n\|^4)\,\left( 2 e^{-2\gamma} \|\EE(W_n)\| + 8 e^{-2\gamma} \EE(\|W_n\|^2) \right)
\end{equation}
In summary we find
$$ \EE(\lVert v_{n+1} \rVert^4)\leq  \EE(\lVert v_n \rVert^4) (1+\alpha(\gamma)||\EE(W_n)||+\beta(\gamma)\EE||(W_n)||^2) $$
where $\alpha(\gamma)$ and $\beta(\gamma)$ are some positive conitnuous functions in $\gamma$.
Taking $\Cc_\gamma = \max(\alpha(\gamma), \beta(\gamma))$ we find 
\begin{align*}
    \EE(\lVert v_{n+1}\rVert^4) &\leq  \lVert v_0 \rVert^4 \prod_{n \geq 1}  (1+\alpha(\gamma)||\EE(W_n)|| +\beta(\gamma)\EE||(W_n)||^2))\\
    &\leq  \lVert v_0 \rVert^4 \exp\left[\Cc_\gamma \left(\sum_{n \geq 1} ||\EE(W_n)|| +\EE||(W_n)||^2\right)\right]\leq  \lVert v_0 \rVert^4 \exp(\Cc_\gamma \Cc_W)
\end{align*}
Use $\|X\| \leq \sum_k \|Xw_k\| $ for $(w_k)_k$ being some orthogonal basis to get the result.
\end{proof}

\section{Applying the key estimates to the transfer matrices}

The main point of this section will be to apply  the estimates from Section~\ref{sec:3-key-est} to the conjugated transfer matrices as developed
in Section~\ref{sec:3-channels}.
Like indicated at the end of Section~\ref{sec:3-channels} we choose some compact interval $[a,b] \subset \Sigma$ such that for $\lambda \in [a,b]$ the first $l_e$ channels are elliptic and the other $l-l_e=l_h$ channels are hyperbolic.
In the notations of the previous sections,  we have $l_0=2l_e+l_h$,  $l_1=l_h$ and the matrices $T$ and $S$ as defined in \eqref{eq-setup-T} are given by
$$
T\,=\, \pmat{S \\ & \Gamma} \;,\qquad  S\,=\,\pmat{e^{-iK} \\ & e^{iK} \\ & & \Gamma^
{-1}}\;.
$$
where $K$ and $\Gamma$ depend analytically on $z=\lambda+i\eta \in [a,b]+i[-c,c]$.
Using continuity and compactness arguments we have uniform estimates like
$$
\|\Gamma^{-1}\|\,<\, e^{-2\gamma}\;,\quad \|Q_z\|<\Cc_Q\;, \quad \|Q_z^{-1}\|<\Cc_Q
$$
for all $z=\lambda+i\eta$ with $\lambda \in [a,b]$ and $\eta \in [-c, c]$\;. This leads to
\begin{equation} \label{eq-V-Vv}
\|\Vv^z_n\|\,=\,\left\|Q_z^{-1} \pmat{V_n &0 \\ 0 & 0} Q_z \right\|\,\leq\, \Cc_Q^2 \,\|V_n\|\,=\,\Cc_Q^2\, \left\|V_n(\omega)\right\|\;.
\end{equation}
for all $z=\lambda+i\eta \in [a,b]+i[-c,c]\subset \CC$. 

In order to apply the results of Section~\ref{sec:3-key-est} we need
$\|\Vv^z_n\|<\frac{e^{2\gamma}-e^\gamma}{4}$.  We therefore will replace $\Vv^\lambda_n$ by
\begin{equation}
W^z_n=W^z_n(\omega)=\Vv^z_n(\omega)\,\cdot\,  1_{\|V_n\|<(e^{2\gamma}-e^\gamma)/(4\Cc_Q^2)}(\omega)\;
\end{equation}
where the latter expression is the indicator function on the event that $\|V_n(\omega)\|<\frac{e^{2\gamma}-e^\gamma}{4\Cc_Q^2}$ on the probability space $\Omega$.
This means essentially to replace the potential $V_n$ by
$$
\widehat V_n\,=\,V_n\,\cdot \,1_{\|V_n\|<\frac{e^{2\gamma}-e^\gamma}{4\Cc_Q^2}}\;.
$$
Note that by the estimates above 
\begin{equation}
\|W^z_n\|<\frac{e^{2\gamma}-e^\gamma}{4} \qtx{for all} z=\lambda+i\eta \in [a,b] + i [-c, c]
\end{equation}

We modify the transfer matrices accordingly and let
\begin{equation}
\widehat{T}^z_n=\pmat{A+\widehat V_n \,-\,z \,I & -I \\ I & \nul}\;.
\end{equation}
With these definitions we note that
$$
Q_z^{-1} \widehat T^z_n Q_z \,=\,T^z\,+\,W^z_n\;.
$$
Similarly to the products $T^z_{m,n}$ we define
$$
\widehat T^z_{m,n}\,=\,\widehat T^z_n \widehat T^z_{n-1} \cdots \widehat T^z_{m+1} \widehat T^z_m\;.
$$
and
\begin{equation}
\Xx^{z}_{m,n}\,=\, Q_z^{-1} \widehat T^z_{m,n} Q_z
\end{equation}
Using the splitting into blocks of sizes $l_0=2l_e+l_h$ and $l_1=l_h$ we write
\begin{equation}
\Xx^z_{m,n}\,=\,\pmat{A^z_{m,n} & B^z_{m,n} \\ C^z_{m,n} & D^z_{m,n} }
\end{equation}
and we define the Schur complements
\begin{equation}
X^z_{m,n}\,=\,A^z_{m,n}\,-\, B^z_{m,n} \left(D^z_{m,n}\right)^{-1}  C^z_{m,n}\qtx{and}
Z^z_{m,n}\,=\,B^z_{m,n} \left(D^z_{m,n}\right)^{-1} \;.
\end{equation}
The reason that we will work with the products from some $m$ on is that for large $n\geq m$ and some random $m$, we will have that $V_n=\widehat V_n$.
More precise probabilistic arguments will be given later.

First,  we need to check that the matrices $W^z_n$ do indeed satisfy the bounds we need:

\begin{proposition}
There exists $\Cc_W<\infty$ (depending on the chosen compact interval $[a,b]\subset \Sigma$ and the chosen $c >0$) such that for all $z=\lambda +i\eta \in [a,b]+i[-c,c]$ we have
$$
\sum_{n=0}^\infty \left( \left\| \EE\left(W^z_n\right) \right\|\,+\,\EE\left( \|W^z_n \|^2\right)\,\right)\,\leq\,\Cc_W\;.
$$
\end{proposition}
\begin{proof}
First we note
$$
\sum_{n=0}^\infty \EE\left(\| W^z_n\|^2\right)\,\leq\,
\sum_{n=0}^\infty \EE\left( \|\Vv^z_n \|^2\right)\,\leq\,
\Cc_Q^4\,\sum_{n=0}^\infty \EE\left( \left\|V_n \right\|^2\right)\,=\,\Cc'\,<\,\infty\;.
$$
uniformly for $z \in [a,b]+i[-c,c]\subset \CC$,  which bounds the second term as needed.

Using the Cauchy-Schwartz Inequality in $L^2(\Omega,\Ff,\bP)$ we find
\begin{align*}
\EE\left(\left\|W^z_n -\Vv^z_n\right\|\right)\,&=\,
\EE\left(\left\|\Vv^z_n \right\| \cdot 1_{\|V_n\|\geq\frac{e^{2\gamma}-e^\gamma}{4\Cc_Q^2}}\right) \\
&\leq\, \sqrt{\EE\left(\left\|\Vv^z_n \right\|^2\right)}\, \sqrt{\EE \left( 1_{\|V_n\|\geq \frac{e^{2\gamma}-e^\gamma}{4\Cc_Q^2}}\right)}\;.
\end{align*}
For the first term we use \eqref{eq-V-Vv},  for the second term, we use Chebyshev's inequality
$$
\EE \left( 1_{\|V_n\|\geq\frac{e^{2\gamma}-e^\gamma}{4\Cc_Q^2}}\right)\,=\,
\bP\left(\|V_n\|\geq \frac{e^{2\gamma}-e^\gamma}{4\Cc_Q^2}\right)\,\leq\, \frac{16\,\Cc_Q^4}{(e^{2\gamma}-e^\gamma)^2} \,\EE(\|V_n\|^2)\;
$$
in order to get
\begin{equation}
\EE\left(\left\|W^\lambda_n -\Vv^\lambda_n\right\|\right)\,\leq\, \frac{4\,\Cc_Q^4}{e^{2\gamma}-e^\gamma}\, \EE(\|V_n\|^2)
\end{equation}
for all $z \in [a,b]+i[-c,c]\subset \CC$.
Thus,
$$
\left\|\EE\left( W^z_n \right)\right\|\,\leq\,
\left\|\EE\left( \Vv^z_n \right)\right\|\,+\,
\EE\left(\| W^z_n-\Vv^z_n\| \right) \,\leq\,
\Cc_Q \left\| \EE\left( V_n \right)\right\|\,+\,\frac{4\,\Cc_Q^4}{e^{2\gamma}-e^\gamma} \EE(\|V_n\|^2)
$$
which leads to
$$
\sum_{n=0}^\infty \left\|\EE\left( W^\lambda_n \right)\right\|\,\leq\,
\sum_{n=0}^\infty \left( \Cc_Q \left\| \EE\left( V_n \right)\right\|\,+\,\frac{4\,\Cc_Q^4}{e^{2\gamma}-e^\gamma} \EE(\|V_n\|^2) \right)\,=\,\Cc''<\infty.
$$
Now $\Cc_W=\Cc'+\Cc''$ does the job.
\end{proof}

Thus, we can apply the results from Section~\ref{sec:3-key-est}.

\noindent 
\begin{proposition}\label{prop-limits+est}
Let $\Omega'=\{\omega\,:\, \lim_{n\to \infty} V_n(\omega)=0\}$ which satisfies $\bP(\Omega')=1$.  

\begin{enumerate}
\item[{\rm (i)}] 
For all $\omega \in \Omega'$,  all $m\in\ZZ_+$,  and 
for all $z \in [a,b]+i[-c,c]$ we have, 
$$
\lim_{n\to\infty}  Z^z_{m,n}=0\;, \quad \lim_{n \to \infty} (D^z_{m,n})^{-1}= 0\;, \qtx{and}
Y_m^z\,:=\, \lim_{n\to \infty} (D^z_{m,n})^{-1} C^z_{m,n} \quad \text{exists.}
$$
\item[{\rm (ii)}] For all $\omega \in \Omega'$,  all $m\in \ZZ_+$
$$
z \mapsto Y_m^\lambda \quad\text{is analytic for $z=\lambda+i\eta \in (a,b)+i(-c, c)$} 
$$
and,  uniformly in $z=\lambda+i\eta \in [a,b]+i[-c,c]$ we find
$$
\lim_{m\to\infty} Y^z_m\,=\,0\;.
$$
\item[{\rm (iii)}] We have for all $\omega \in \Omega$,  and $z=\lambda+i\eta \in [a,b] +i [-c,c]$ that
$\|Z^z_{m,n}\| \leq 1\;.$
\item[{\rm (iv)}] We find $\Cc>0$ such that (uniformly) for all $\lambda \in [a,b]$ and all $m\in \ZZ_+$ 
$$
\sup_{n\geq m} \EE(\|X^\lambda_{m,n}\|^4)\,\leq\,\Cc<\,\infty\;.
$$
\end{enumerate}
\end{proposition}

\begin{proof}  For part (i) note that with probability 1,  $\|V_n\| \to 0$ for $n \to \infty$.
We let $\Omega' \subset \Omega$ be the set of probability one where $V_n=V_n(\omega) \to 0$.
Then we have the same for $W^z_n$ and the limits follow from Proposition~\ref{prop-limits}.

For part (ii) note first that $(D_{m,n}^z)^{-1} C_{m,n}^z$ is analytic in $z \in [a,b]+i[-c,c]$.
Now,  for $\omega \in \Omega'$ fixed,  one sees from the estimates in Proposition~\ref{prop-limits} that the convergence of the series $\sum_{n>m} [(D_{m,n+1}^z)^{-1} C_{m,n+1}^z-(D_{m,n}^z)^{-1} C_{m,n}^z]$ is uniform for $z \in [a,b]+i[-c, c]$.  Hence,  the limiting function is analytic in $z$.
Moreover,  if for all $n>m$ we have  $\|V_n\|<\varepsilon$ then one sees that with a uniform constant $\Cc_Y<\infty$,  we have $\|Y^z_m\|<\varepsilon \Cc_Y$.
As $V_n \to 0$ for $\omega \in \Omega'$,  we find $\varepsilon$ arbitrarily small as $m\to \infty$ and hence $\lim_{m\to \infty} Y^z_m=0$ uniformly in $z$.

Part (iii) simply follows from Proposition~\ref{prop-boundZ} and part (iv) from Proposition~\ref{prop-boundX},  noting that all bounds are uniform for $\lambda\in [a,b]$ and $\|S\|\leq 1$ for $\lambda \in [a,b]$.
\end{proof}

\begin{proposition}\label{prop-est-schur}
There is a set of probability one,  $\tilde \Omega\subset \Omega$,  $\bP(\tilde \Omega)=1$,  such that for any $\omega \in \tilde \Omega$ and any $m\in \ZZ_+$ we find
$$
\liminf_{n\to \infty} \int_a^b \| X^\lambda_{m,n}\,\|^4\, {\rm d}\lambda\,< \, \infty\;.$$
\end{proposition}

\begin{proof}
By Proposition~\ref{prop-limits+est}~(iv),  Fatou lemma and Fubini theorem we find
$$\EE \,\liminf_{n\to \infty} \int_a^b (||X^z_{m,n}||^4) d\lambda \leq\liminf_{n\to \infty}  \int_a^b  \EE ( ||X^z_{m,n}||^4) d\lambda \leq \Cc(b-a)\,<\,\infty $$
Hence,  $\bP\left(\liminf_{n\to\infty} \int_a^b \|X^z_{m,n}\| {\rm d}\lambda = \infty\right)=0$.
\end{proof}

\section{Absolutely continuous spectrum}

In this section we finally prove Theorem~\ref{th:main}.
Recall in Proposition~\ref{prop-limits+est} we defined the set $\Omega'$ of probability one,  where $V_n \to 0$.
For $\omega \in \Omega'$ we find $m$ such that for $n\geq m$  and all $\lambda \in [a,b]+i[-c,c]$ we have $\Vv^z_n=W^z_n$.  
However,  the $m$ is random and not uniform in $\omega$.  Therefore,  
we define the events
$$
\Omega_m\,=\,\left\{\omega \in \Omega' \,:\,\left(  \forall n\geq m\,:\, \|V_n(\omega)\|<\frac{e^{2\gamma}-e^\gamma}{4\Cc_Q^2}\;\right)\;\right\}
$$
For $\omega \in \Omega_m$, $z=\lambda+i\eta\in[a,b]+i[-c,c]$ and $n\geq m$  we find  $\widehat T^z_n=T^z_n$ and,  hence, 
$$
T^z_{m,n}(\omega)\,=\, \widehat T^z_{m,n}(\omega)\,=\,  Q_z \Xx^z_{m,n}(\omega) Q_z^{-1} \;.
$$
Moreover,
$$
\bP\left( \bigcup_{m=0}^\infty \Omega_m\right)\,=\,\bP\left( \Omega'  \right)\,=\,1.
$$

The main work left to do now is to use Proposition~\ref{prop-est-schur} to obtain an estimate of the form as needed in Theorem~\ref{th-ac-criterion} which is a special case of \cite[Theorem~4]{Sadel2019}.
Thus,  given a vector $\vx \in \CC^{l}$ associated to some vector in the 0-th shell, we need to find vectors $\vec{u}_{\lambda,n} \in \CC^l$ such that
$$
\liminf_{n\to\infty} \int_a^b \left\|T^\lambda_{0,n} \smat{\vec{u}_{\lambda,n} \\ \vx} \right\|^4\,{\rm d}\lambda\,<\, \infty\;.
$$
Note that for $\omega \in \Omega_m$ one has
$$
T^z_{0,n} \pmat{\vec{u}_{z,n} \\ \vx} \,=\, Q_z \pmat{A^z_{m,n} & B^z_{m,n} \\ C^z_{m,n} & D^z_{m,n}}
Q_z^{-1} T^z_{0,m-1}  \pmat{\vec{u}_{z,n} \\ \vx}\;.
$$

\begin{lemma}\label{lem:u-y}
For $Y \in \CC^{l_h \times (l+l_e)}$ assume 
\begin{equation}\label{cond:rank}
{\rm rank}\,\left[  \pmat{Y & I_{l_h} } Q_\lambda^{-1} T^\lambda_{0,m-1} \pmat{I_l \\ 0} \right]\,=\,l_h\;,
\end{equation}
then,  for any $\vx$ one finds $\vec{u}_{\lambda,Y} \in \CC^l$, $\vec{y}_Y \in \CC^{l+l_e}$ such that
\begin{equation}\label{eq-u-y-general}
Q_\lambda^{-1} T^\lambda_{0,m-1} \pmat{\vec{u}_{\lambda,Y} \\ \vx}\,=\,
\pmat{\vec{y}_{\lambda,Y} \\ -Y \, \vec{y}_{\lambda,Y} }\;.
\end{equation}
If the condition \eqref{cond:rank} is fulfilled for specific $\lambda=\lambda_0$ and $Y=Y_0$,  then it is fulfilled in a neighborhood of $(\lambda_0,Y_0)$. Moreover,  given a fixed vector $\vx$,  one may get solutions $\vec{u}_Y$ and $\vec{y}_Y$ that depend continuously on $(\lambda,Y)$ in a neighborhood of $(\lambda_0,Y_0)$. 

Now let $\omega \in \Omega_m$ and use $Y=(D^\lambda_{m,n})^{-1} C^\lambda_{m,n}$ and denote  $\vec{u}_{\lambda,Y}$, $\vec{y}_{\lambda,Y}$ by$\vec{u}_{\lambda,n}$ and $\vec{y}_{\lambda,n}$. Hence, 
\begin{equation}\label{eq-u-y}
T^\lambda_{0,n} \pmat{\vec{u}_{\lambda,n} \\ \vx}\,=\,Q_\lambda \pmat{ X^\lambda_{m,n} \;\vec{y}_{\lambda,n} \\ 0}\;.
\end{equation}
\end{lemma}

\begin{proof}
Using \eqref{eq-u-y-general} in the decomposition of $T^z_{0,n}$ above,  with $z=\lambda$ and $Y=(D^\lambda_{m,n})^{-1} C^\lambda_{m,n}$,  the statement \eqref{eq-u-y} follows directly.
Thus, we need to check that we find $\vec{u}_{\lambda,Y}$ and $\vec{y}_{\lambda,Y}$ such that \eqref{eq-u-y-general} is satisfied.
Dividing the $2l \times 2l$ matrix $Q_\lambda^{-1} T^\lambda_{0,m-1}$ horizontally in blocks of sizes $l$ and $l$, and vertically into blocks of sizes $l+l_e$ and $l_h$ we may write
$$
Q_\lambda^{-1}T^\lambda_{0,m-1} = \pmat{\aaa & \bbb \\ \ccc & \ddd} \qtx{and}
Q_\lambda^{-1}T^\lambda_{0,m-1}  \pmat{\vec{u}_{\lambda,n} \\ \vx}\,=\,
\pmat{\aaa \vec{u}_{\lambda_n} + \bbb \vx \\ \ccc \vec{u}_{\lambda,n} + \ddd \vx}
$$
Note $\aaa, \bbb \in \CC^{(l+l_e) \times l},\,\ccc, \ddd \in \CC^{l_h \times l}$\;.
Then,  \eqref{eq-u-y-general} is satisfied for $\vec{y}_{\lambda,Y}=\aaa \vec{u}_{\lambda,Y} + \bbb \vx$ if and only if
$$
\ccc \vec{u}_{\lambda,Y} + \ddd \vx\,=\, - Y \,\left( \aaa \vec{u}_{\lambda,Y} + \bbb \vx \right)
$$
This is equivalent to
$$\left(Y \, \aaa+\ccc\right) \vec{u}_{\lambda,Y}\,=\,
\left(-Y \bbb - \ddd \right)\,\vx\;.
$$
Thus,  we find a solution $\vec{u}_{\lambda,Y}$ for any $\vx$, if $Y \, \aaa+\ccc$ is surjective (as a linear map from $\CC^l$ to $\CC^{l_h}$),
which is exactly the rank condition given in the assumption.

Note, if this is fulfilled for some specific $Y=Y_0$, and some specific spectral parameter $\lambda=\lambda_0$, 
then we find a matrix $M\in \CC^{l \times l_h}$ such that
$\det((Y_0 \, \aaa+\ccc) M)\neq 0$,. So in a neighborhood of $Y_0$ and $\lambda_0$,  this determinant is still not zero and we may use
$$
\vec{u}_Y\,=\,M[(Y\aaa+\ccc)M]^{-1} \,\left( -Y\bbb-\ddd\right)\,\vx\;
$$
and as above,  $\vec{y}_{\lambda,Y}=\aaa \vec{u}_{\lambda,Y} + \bbb \vx$.  Thus,  both
depend continuously on $(\lambda,Y)$. 
\end{proof}

In the sequel need to use the form of the transfer matrices as in \cite{Sadel2019} using the resolvent boundary data of restrictions to finite graphs.
Thus, let $H_{0,m}=H_{0,m}(\omega)$ be the restriction of $H_\omega$ to $\ell^2(\{0,\ldots,m\}) \otimes \CC^l$,  that is
$H_{0,m}(\omega)= P^* H_\omega P$ where $P:\ell^2(\{0,\ldots,m\}) \otimes \CC^l \hookrightarrow \ell^2(\ZZ_+)\otimes \CC^l$ is the natural
embedding. Note that $H_{0,m}$ is a $l(m+1) \times l(m+1)$ Hermitian matrix.
One may define the resolvent boundary data from shell 0 to $m$ as in \cite{Sadel2019} by
\begin{equation}\label{eq-boundary-data}
\pmat{\alpha^z_{0,m} & \beta^z_{0,m} \\ \gamma^z_{0,m} & \delta^z_{0,m} }=
\pmat{P_0^* \\ P_m^*} (H_{0,m}-z)^{-1} \pmat{P_0 & P_m}
\end{equation}
where $P_k$ is the natural embedding of $\ell^2(\{k\})\otimes \CC^l$ into $\ell^2(\{0,\ldots,m\})\otimes \CC^l$ and can be regarded as an $l(m+1) \times l$ matrix. 
In this sense, 
$$
P_0=\pmat{I_l \\ \nul \\ \vdots \\ \nul}, \qquad P_m= \pmat{\nul \\ \vdots \\ \nul \\ I_l}
$$
and $\alpha^z_{0,m}, \beta^z_{0,m}, \gamma^z_{0,m}$ and $\delta^z_{0,m}$ are all $l \times l$ matrices.
Then, one of the main points following from the work in \cite{Sadel2019} is the following formula, which we also prove in Appendix~\ref{app-Transfermatrix}.

\begin{proposition}\label{prop-boundary-data}(cf. Proposition~\ref{prop-app-A})
If $z$ is not an eigenvalue of $H_{0,m}$ and $\beta^z_{0,m}$ is invertible,  then
$$
T^z_{0,m} = \pmat{(\beta^z_{0,m})^{-1} & -(\beta^z_{0,m})^{-1} \alpha^z_{0,m} \\ \delta^z_{0,m} (\beta^z_{0,m})^{-1} &
\gamma^z_{0,m} - \delta^z_{0,m} (\beta^z_{0,m})^{-1} \alpha^z_{0,m} }
$$
\end{proposition}

Now, we can continue with the following.  Note, that $\omega\in \Omega_{m'}$ and $m\geq m'$ implies $\omega\in \Omega_m$.

\begin{lemma}\label{lem:fullrank}
Given $\omega \in \Omega_{m'}$, and $c> \eta > 0$ fixed,  there exists $\tilde m> m'$ such that
$\forall m > \tilde m$ and $ \forall \lambda \in [a,b]\;:\;$ we have ${\rm rank}\,\left[\Aa^{\lambda+i\eta}_m \right]\,=\,l_h$, where
$$
\Aa^z_m\,:=\,  \pmat{Y^z_m & \;I_{l_h} } Q_{z}^{-1} T^{z}_{0,m-1} \pmat{I_l \\ 0} 
$$
\end{lemma}

\begin{proof} For notation we let $z=\lambda+i\eta$.  From Proposition~\ref{prop-limits+est} part (ii)we find that $Y^z_m$ is uniformly small for $m$ sufficiently big.  This means,  for any $\varepsilon>0$,  there exists $\tilde m> m'$ such that for any $n>m>\tilde m$ and any
$z=\lambda+i\eta \in [a,b]+i[-c,c]$ we have
$$
\| Y^z_{m}\|\,<\, \varepsilon\;.
$$
The $\varepsilon$ needed for the statement will be chosen later.

Using the definitions \eqref{eq-boundary-data} and Proposition~\ref{prop-boundary-data}
we find
$$
T^z_{0,m-1} \pmat{I_l \\ 0}\,=\,\pmat{(\beta^z_{0,m-1})^{-1} \\ \delta^z_{0,m-1} (\beta^z_{0,m-1})^{-1}}
$$
By the other ways of writing the transfer matrix, we see that $(\beta^z_{0,m-1})^{-1}$ exists for any $z$, at least after analytic continuation.
We also note that $\beta^z_{0,m-1}$ exists for any value $z$ except for the eigenvalues of $H_{0,m-1}$. Thus,  it exists for any $z=\lambda+i\eta$ with $\eta>0$.

In order to prove that $\Aa^z$ is of full rank $l_h$, it is sufficient to prove that $\Qq_\Gamma^{-1} \Gamma \Aa^z \Bb$ is invertible,  where
$\Bb\in \CC^{l \times l_h}$.
In particular,  we consider
$$
\Bb'=\beta^z_{0,m-1} \pmat{ 0 \\ I_{l_h}}\qtx{giving} \Aa^z_m \Bb'\,=\,
\pmat{Y^z_{m} & I_{l_h}} Q_z^{-1} \pmat{I_l \\ \delta^z_{0,m-1}} \pmat{0 \\ I_{l_h}}
$$
First,  take the 'limit case' and with \eqref{eq-Q^-1} we find
$$
\pmat{0 & I_{l_h}} Q_z^{-1} \pmat{I_l \\ \delta^z_{0,m-1}} \pmat{0 \\ I_{l_h}}\,=\,\pmat{0 & -\Qq_\Gamma & 0 & \Gamma^{-1} \Qq_\Gamma}
\pmat{0 \\ I_{l_h} \\ \delta^z_{0,m-1} \pmat{0 \\ I_{l_h}}}
$$
$$\;=\;
\Qq_\Gamma\Gamma^{-1} \left( -\Gamma+ \pmat{0 &  I_{l_h}} \delta^z_{0,m-1} \pmat{0 \\ I_{l_h}} \right)\;. 
$$
were we note that by their definition,  $\Qq_\Gamma=(\Gamma^{-1}-\Gamma)^{-1}$ and $\Gamma$ commute.
Thus, we find
$$
\Qq_\Gamma^{-1} \Gamma \Aa^z_m \Bb\,=\,
\pmat{0 &  I_{l_h}} \delta^z_{0,m-1} \pmat{0 \\ I_{l_h}}-\Gamma+\Rr^z
$$
where
$$
\Rr^z\,=\,\Qq_\Gamma^{-1} \Gamma \pmat{ Y^z_{m} & 0}Q_z^{-1} \pmat{I_{l_h} \\ \delta^z_{0,m-1}} \pmat{0 \\ I_{l_h}}\;.
$$
Using $\| \delta^z_{m,n} \| \leq \frac{1}{\eta}$ ,  where $z=\lambda+i\eta$,  and compactness, we get with some uniform constant $\Cc>0$ that
$$
\|\Qq_\Gamma^{-1} \Gamma\| \|Q_z^{-1}\|\,<\, \Cc \qtx{and} \left\| \pmat{I_l\\ \delta^z_{0,m-1} }\right\|\,<\, 1+ \frac1 \eta\;
$$
for all $z\in[a,b]+i[-c,c]$ and all $m> \tilde m$.  This gives
$$
\|\Rr^z\|\,\leq\, \Cc\,\varepsilon\, \left(1+\frac{1}{\eta}\right)
$$
for any $z=\lambda+i\eta \in [a,b]+i[-c,c]$ and any $m>\tilde m=\tilde m (\varepsilon)$.
Note,  $\Gamma=\Gamma(z)$ is a diagonal matrix,  such that
$$
\Gamma+\Gamma^{-1}\,=\,\pmat{\alpha_{l_e+1}-z \\ & \ddots \\\ & & \alpha_l-z}
$$
Moreover, as set above,  all diagonal entries of $\Gamma(z)$ are bigger than $e^{2\gamma}>1$.
We note, that the imaginary parts of $\Gamma^{-1}$ have opposite sign and  an absolute value smaller than for the corresponding values of $\Gamma$.
Thus, we find for $\eta>0$ that
$$
\Im(-\Gamma-\Gamma^{-1})=\eta \,I\, \qtx{implying} \Im(-\Gamma)>\eta \,I\;.
$$
In general, we will define the "imaginary" part in $C^*$ algebra sense, that is
$\Im(\Aa)=(\Aa-\Aa^*)/(2i)$,  
then 
$$
\Im\left[ \pmat{0 & I_{l_h}} \delta^z_{0,m-1} \pmat{0 \\ I_{l_h}}\right]\,>\,0 
$$
for $\eta>0$ and $z=\lambda+i\eta$ Hence,  we finally obtain
$$
\Im\left(\Qq_\Gamma^{-1} \Gamma \Aa^z_m \Bb - \Rr^z \right)\,>\, \eta\,I_{l_h}\,.
$$
Thus,  if
$$
\varepsilon < \frac{\eta^2}{\Cc(1+\eta)}  \qtx{implying} \|\Rr^z\|\,<\,\eta
$$
then,  the $l_h \times l_h$ matrix $\Qq_\Gamma^{-1} \Gamma \Aa^z_m \Bb$ is invertible and we have ${\rm rank}(\Aa^z_m) \geq l_h$ (for any $m>\tilde m$).
By the dimensions of $\Aa^z_m \in \CC^{l_h \times l}$ we also have ${\rm rank} (\Aa^z_m) \leq l_h$.
\end{proof}

\begin{proposition}\label{th:main-part}
Let $\omega \in \Omega'\cap \tilde \Omega$, where $\tilde \Omega$ is the set as in Proposition~\ref{prop-est-schur}.  
Then,  there is a finite set $\{\lambda_1,\ldots, \lambda_k\}$ such that the spectrum of $H_\omega$ is purely absolutely continuous in
$(a,b) \setminus \{\lambda_1,\ldots, \lambda_k\}$.  \\
If there is no hyperbolic channel, that is $l_h=0$,  then the spectrum is purely absolutely continuous in $(a,b)$.
\end{proposition}

\begin{proof}
For some $m'$ we find $\omega \in \Omega_{m'}$.  Choose $\eta$ with $c>\eta>0$,  take $\tilde m>m'$ as in Lemma~\ref{lem:fullrank} and consider some fixed $m>\tilde m$. 
We note that we also have $\omega \in \Omega_m$.  Now,  using the notation as above,  $\Aa^z_m$ has full rank for $\Im m(z)= \eta$. 
By analyticity,  the rank of $\Aa^z_m$ is full for all but finitely many values of $z=\lambda+i\eta \in [a,b]+i [-c,c]$. 
We may now restrict to the real line again and let
$\{\lambda_1, \ldots, \lambda_k\} \subset [a,b]$ be the finite set of energies, where ${\rm rank}(\Aa^\lambda_m)< l_h$.

We consider now a compact interval $[a',b']\subset [a,b] \setminus\{\lambda_1,\ldots,\lambda_k\}$.
For all $\lambda \in [a',b']$ we find that $\Aa^\lambda_m$ has full rank $l_h$.
By compactness,  the set $\{\Aa^\lambda_m\,:\, \lambda \in [a',b']\}$ has some positive distance,  say $\varepsilon>0$,  to the set of $l_h \times l$ matrices of non full rank.

In order to get to the point of Lemma~\ref{lem:fullrank}, let us introduce the notations
$$
\Aa^z_{Y} = \pmat{ Y & I_{l_h} } Q_z^{-1} T^{z}_{0,m-1} \pmat{I_l \\ 0}
$$
$$
Y^z_{m,n} = (D^z_{m,n})^{-1} C^z_{m,n}\qtx{and}
\Aa^z_{m,n} = \Aa^z_{Y^z_{m,n}}= \pmat{ Y^z_{m,n} & I_{l_h} } Q_z^{-1} T^{z}_{0,m-1} \pmat{I_l \\ 0}\;.
$$
Again by compactness we note that $\|Q_\lambda^{-1} T^\lambda_{0,m-1} \|<\Cc$ for all $\lambda \in [a,b]$.  (Note, that $m$ is fixed now!).
Thus we see that
$$
\|\Aa^\lambda_{Y} - \Aa^\lambda_m \| \,\leq\, \Cc \|Y - Y^\lambda_m\|
$$
for all $\lambda \in [a,b] \supset [a',b']$.
Therefore,  if 
$$
\|Y-Y^\lambda_m\| < \frac{\varepsilon}{\Cc} \qtx{implying}  \|\Aa^\lambda_Y - \Aa^\lambda_m\|<\varepsilon
$$
then $\Aa^\lambda_Y$ is of full rank $l_h$.

Now,  consider the compact set
$$
\Ss=\left\{\ ( \lambda,Y)\,:\, \lambda\in[a',b']\,,\,\|Y-Y^\lambda_m\|\leq  \frac\varepsilon{2\Cc} \right\}\;.
$$
By Lemma~\ref{lem:u-y},  for any $(\lambda',Y') \in \Ss$, we find some  neighborhood $\Uu_{\lambda',Y'}$ and solutions $\vec{u}_{\lambda,Y}$, $\vec{y}_{\lambda,Y}$ to \eqref{eq-u-y-general},  that depend continuously on $(\lambda,Y) \in \Uu_{Y', \lambda'}$. Possibly shrinking the neighborhood a bit, we may assume it is compact,
and thus,  $\|\vec{y}_{\lambda,Y}\|$ attains a maximum in $\Uu_{\lambda',Y'}$.
By compactness,  $\Ss$ can be covered by finitely many such compact neighborhoods $\Uu'$.
Making a specific choice in the overlaps of these finitely many neighborhood,  we find piece-wise continuous functions
$$
\vec{u}\,:\, \Ss \to \CC^l\;, \;(\lambda,Y) \to \vec{u}_{\lambda,Y}\;, \quad \vec{y}\,:\, \Ss \to \Cc^{l+l_e}\; ,\;
(\lambda,Y) \to \vec{y}_{\lambda,Y}
$$
satisfying equation \eqref{eq-u-y-general} such that for some constant $\Cc_{\vec{y}} < \infty$ and all $(\lambda,Y) \in \Ss$ we have
$$
\| \vec{y}_{\lambda,Y}\|\,\leq\, \Cc_{\vec{y}}\;.
$$

As mentioned in the proof of Proposition~\ref{prop-limits+est} part (ii), the convergence of $Y^z_{m,n} \to Y^z_m$ for $n\to \infty$ is uniform in $z$,  as such
we find $N>0$ such that $\forall n>N$ and all $\lambda\in[a',b']$ we have
$$
\|Y^\lambda_{m,n}-Y^\lambda_m\|\,\leq\,\frac \varepsilon{2\Cc} \qtx{implying} (\lambda, Y^\lambda_{m,n})\,\in\,\Ss\;.
$$
Thus,  for all $n>N$,  and all $\lambda \in [a',b']$ we may choose
$$
\vec{u}_{\lambda,n}= \vec{u}_{Y^\lambda_{m,n},\lambda}\;, \quad \vec{y}_{\lambda,n}=\vec{y}_{Y^\lambda_{m,n},\lambda}\;.
$$
By \eqref{eq-u-y} we obtain that
\begin{equation}
\left\|T^\lambda_{0,n} \smat{\vec{u}_{\lambda,n} \\ \vx} \right\|\,=\,
\left\|Q_\lambda \pmat{ X^\lambda_{m,n}\, \vec{y}_{\lambda,n} \\ 0} \right\| \, \leq\,  \|Q_\lambda\|\, \|X^\lambda_{m,n}\|\, \|\vec{y}_{\lambda,n}\|\, \leq \,
\Cc_Q \Cc_{\vec{y}} \| X^\lambda_{m,n}\|
\end{equation}
for $\lambda \in [a',b']$ and all $n>N$.

Hence,  using that $\omega \in \tilde \Omega$ we get by Proposition~\ref{prop-est-schur} that
$$
\liminf_{n\to \infty} \int_{a'}^{b'} \left\|T^\lambda_{0,n} \smat{\vec{u}_{\lambda,n}\\ \vx} \right\|^4\,{\rm d}\lambda\,\leq\,
\Cc_Q \Cc_y \liminf_{n\to\infty} \int_{a'}^{b'} \|X^\lambda_{m,n}\|^4\,{\rm d}\lambda\,<\, \infty\;.
$$
Hence,  Theorem~\ref{th-ac-criterion} gives that the spectral measure at $\delta_0 \otimes \vx$ is purely absolutely continuous in $[a',b']$.  
As $\vx \in \CC^l$ was arbitrary,  and the closures of ${\rm span}(\{(H_\omega)^k \delta_0 \otimes \vx\,:\, \vx \in \CC^l,\; k\in \NN_0\})$ is the whole Hilbert space,
we find that the spectrum of $H_\omega$ is purely absolutely continuous in $(a',b')$.
Now,  the set $\Sigma'=[a,b] \setminus\{\lambda_1,\ldots,\lambda_k\}$ can be written as countable union of intervals $(a',b')$ such that $[a',b'] \subset \Sigma'$.
Therefore,  the spectrum of $H_\omega$ is purely absolutely continuous in $\Sigma'$. \
Note that $\lambda_1, \ldots, \lambda_k$ may be eigenvalues of $H_\omega$,  but they do not have to be. If $\lambda_j$ is not an eigenvalue, 
then the spectrum of $H_\omega$ is also purely absolutely continuous in a neighborhood of $\lambda_j$. 
Thus, we only need to subtract the eigenvalues from the set $[a,b]$.\\
Note,  in the intersection of all the bands,  that is,  if $l_h=0$,  one has $X^\lambda_{m,n}= Q_\lambda^{-1} T^\lambda_{m,n} Q_\lambda$,  $\vec{y}_{\lambda,n} =Q_\lambda^{-1} T^\lambda_{0,n} \smat{\vec{u}_{\lambda,n} \\ \vx}$ and one can choose any family of uniformly bounded vectors $\vec{u}_{\lambda,n}$ to get pure absolutely continuous spectrum in $(a,b)$. There is no need to subtract a finite set of values.
\end{proof}

Theorem~\ref{th:main} now essentially follows directly from this proposition:

\begin{proof}
First note that the set $\Omega'$ does not depend on the interval $[a,b]$ analyzed above,  but $\tilde \Omega$ does. 
Using compact intervals inside $\Sigma$ with rational boundary points we may write $\Sigma$ as countable union of open intervals, whose closure is inside $\Sigma$,
$$
\Sigma=\bigcup_{i=1}^\infty (a_i, b_i) \qtx{where} [a_i, b_i] \subset \Sigma\;.
$$
As $\Sigma$ does not contain any band-edges,  for each $j=1\ldots, l$ the type of the $j$-th channel does not change in $[a_i, b_i]$.  Therefore,  one can make the whole analysis as done for the compact interval $[a,b]$ above for the interval $[a_i,b_i]$.  In particular,  there is a corresponding set $\tilde \Omega_j$ of probability one for the set $[a_i, b_i]$.
We then let
$$
\hat \Omega = \Omega' \cap \bigcap_{i=1}^\infty \tilde \Omega_i 
$$
and note $\bP(\hat \Omega)\,=\,1\;$.
Let $\omega \in \hat \Omega$ and let $\Ccc \subset \Sigma$ be compact.
Using compactness,  there is a finite sub-collection of these intervals,  $[a_{i_k}, b_{i_k}]$, $k=1,\ldots n$, such that
$$
\Ccc \subset \bigcup_{k=1}^n (a_{i_k}, b_{i_k})\;.
$$
Theorem~\ref{th:main-part} gives that there is a finite set $\Eee_k$ of eigenvalues,  such that the spectrum of $H_\omega$ in $(a_{i_k},b_{i_k}) \setminus \Eee_k$ is purely absolutely continuous.
Letting $\Eee=\bigcup_{k=1}^n \Eee_k$,  which is finite,  we see that the spectrum in $\Ccc \setminus \Eee$ is purely absolutely continuous.\\
Due to the last comment in Theorem~\ref{th:main-part}, the spectrum of $H_\omega$ is purely absolutely continuous in the intersection of all bands $\Sigma_0$ (which might be an empty set).
\end{proof}

\section{Acknowledgement}

This work has been supported by the Chilean grants FONDECYT Nr. 1161651,  FONDECYT Nr. 1201836 and the Nucleo Mileneo MESCD.

\appendix

\section{Transfer matrices and spectral averaging formula on the strip \label{app-Transfermatrix}}

As above we consider operators of the form
$$
(H\Psi)_n=-\Psi_{n-1}-\Psi_{n+1}+B _n \Psi_n
$$
on $\ell^2(\ZZ_+)\otimes \CC^l$.
Solving the eigenvalue equation $H\Psi=z\Psi$ leads to the transfer matrices
$$
T^z_n=\pmat{B_n-zI & -I \\ I & \nul} \qtx{and the equation} \pmat{\Psi_{n+1} \\ \Psi_n}=T^z_n \pmat{\Psi_n \\ \Psi_{n-1}}
$$
Then,  for $n>m$ we define the products
$$
T^z_{m,n}=T^z_n T^z_{n-1} \cdots T^z_{m+1} T^z_m \qtx{leading to} \pmat{\Psi_{n+1} \\ \Psi_n}=T^z_{m,n} \pmat{\Psi_m \\ \Psi_{m-1}}
$$
for a formal solution of $H\Psi=z\Psi$.

\subsection{Transfer matrices and resolvent boundary data}

Let $n>m$, be non-negative integers.  With $H_{m,n}$ we denote the restriction of $H$ to $\ell^2(\{m,m+1,\ldots,n\}) \otimes \CC^l$, that is
$$
H_{m,n}=\pmat{B_m & -I & \\ -I & B_{m+1} & -I \\ & \ddots & \ddots & \ddots \\ & & \ddots & \ddots & -I \\ & & & -I & B_n }
$$
Then we define the $m$ to $n$ boundary resolvent data for $z \not \in \sigma(H_{m,n})$ by
\begin{equation}\label{eq-projections}
\pmat{\alpha^z_{m,n} & \beta^z_{m,n} \\ \gamma^z_{m,n} & \delta^z_{m,n}}=
\pmat{P_m^* \\ P_n^*} (H_{m,n}-z)^{-1} \pmat{P_m & P_n}
\end{equation}
where $P_k$ is the natural embedding of $\ell^2(\{k\})\otimes \CC^l$ into $\ell^2(\{m,m+1,\ldots,n\})\otimes \CC^l$ for $m\leq k \leq n$.
This means, e.g.  $\alpha^z_{m,n}=P_m^* (H_{m,n}-z)^{-1} P_m$, and in this setup
$$
P_m=\pmat{I\\ \nul \\ \vdots \\ \nul}\, ,\quad P_n=\pmat{\nul \\ \vdots \\ \nul \\ I}\quad \in\;\;\CC^{(n-m+1)l \times l}\;.
$$
Note that $\alpha^z_{m,n}, \beta^z_{m,n}, \gamma^z_{m,n}, \delta^z_{m,n}$ are all $l\times l$ matrices.

\begin{proposition}\label{prop-app-A}
Let be given $n\geq m\in \ZZ_+$ and let $z\not \in \sigma(H_{m,n})$ and let $\beta^z_{m,n}$ be invertible. 
Then, 
$$
T^z_{m,n}\,=\, \pmat{(\beta^z_{m,n})^{-1} & - (\beta^z_{m,n})^{-1}\alpha^z_{m,n}\\
\delta^z_{m,n} (\beta^z_{m,n})^{-1} & \gamma^z_{m,n} - \delta^z_{m,n} (\beta^z_{m,n})^{-1} \alpha^z_{m,n}}
$$
\end{proposition}

\begin{proof}
For $\Psi=(\Psi_n)_n$ with $\Psi_n \in \CC^l$ we define the notations:
$$\hat \Psi_k := \Psi_k , \quad k<m. \quad \quad \quad \hat \Psi_m :=  \smat{\Psi_m \\ \Psi_{m+1} \\   \vdots \\ \Psi_n }, \qquad
 \hat \Psi_{m+1} : = \Psi_{n+1},$$
and we use $P_m$ and $P_n$ as in \eqref{eq-projections},
then we have
$$\Psi_{m}=P_m^* \hat \Psi_m\;, \quad \Psi_n:=P_n^* \hat \Psi_m $$
and we get
$$\widehat{(H\Psi)}_m:=H_{m,n}\hat \Psi_m-P_m \hat \Psi_{m-1}-P_n \hat \Psi_{m+1}\,.  $$

With $z$ being the spectral parameter, $\widehat{(H \Psi)}_m = z \hat \Psi_m$ leads to
$$P_n \hat \Psi_{m+1} = (H_{m,n}-z)\hat \Psi_m-P_m \hat \Psi_{m-1}$$
Multiplying with $P_m^*(H_{m,n}-z)^{-1}$ from the left, noting that $\hat \Psi_{m+1}=\Psi_{n+1}$ and using \eqref{eq-projections} gives
$$\beta^z_{m,n}  \Psi_{n+1}  =\Psi_m-\alpha^z_{m,n} \Psi_{m-1} \implies  \Psi_{n+1}  =(\beta^z_{m,n})^{-1}\Psi_m-(\beta^z_{m,n})^{-1}\alpha^z_{m,n} \Psi_{m-1} $$
Multiplying from the left  with $P_n^*(H_{m,n}-z)^{-1}$ instead of $P_m^*(H_{m,n}-z)^{-1})$ leads to
$$\delta^z_{m,n} \Psi_{n+1}=\Psi_n-\gamma^z_{m,n} \Psi_{m-1}  $$
Replacing $\Psi_{n+1}$ with the formula above and resolving for $\Psi_n$ leads to
$$\Psi_n=\delta^z_{m,n} (\beta^z_{m,n})^{-1}\Psi_m+(\gamma^z_{m,n} -\delta^z_{m,n} (\beta^z_{m,n})^{-1} \alpha^z_{m,n})\Psi_{m-1} $$
Finally,  we have:
$$\pmat{\Psi_{n+1} \\ \Psi_{n}}=\pmat{(\beta^z_{m,n})^{-1} & - (\beta^z_{m,n})^{-1}\alpha^z_{m,n}\\
\delta^z_{m,n} (\beta^z_{m,n})^{-1} & \gamma^z_{m,n} - \delta^z_{m,n} (\beta^z_{m,n})^{-1} \alpha^z_{m,n}} \pmat{\Psi_m \\ \Psi_{m-1}}
$$ 
As $\Psi_m, \Psi_{m-1}$ determine the solution to $H\Psi=z\Psi$ uniquely,  the matrix must be
 $T^z_{m,n}$.
\end{proof}

\subsection{Spectral averaging formula}

Here we state the strip-equivalent of the spectral average formula from Carmona-Lacroix \cite[Theorem III.3.2 and III.3.6]{Cala}.
It is a special case of \cite[Theorem~1]{Sadel2019}.
First, we need to fix a vector in the root-slice. Thus, we choose some
$\vx \in \CC^l$ which we identify with $\delta_0 \otimes \vx \in \ell^2\{\Z_+\} \otimes \CC^l$.
 Let us assume that $\|\vx\|=1$, so that ${\vx}^* \vx = 1$. Furthermore, identifying $\vx\,^*$ with a linear map from $\CC^l$ to $\CC$, we have a $l-1$ dimensional kernel consisting of the vectors orthogonal to $\vx$,
 $$
 \KK:=\ker(\vx\,^*) = \{\vv \in \CC^l\,:\,\vx\,^* \vv = 0  \} \,=\,\{\vv \in \CC^l\,:\, \vx \cdot \vv = 0\}.
$$
Then, in this special case, the work of \cite{Sadel2019} simply replaces $T^z_0$ by the set of $2l \times 2$ matrices
\begin{equation}
\TT^z_0\,=\, \left\{ \pmat{(B_n-zI) (\vx+\vv) & - \vx +(B_n-zI) \vw\\ \vx+\vv & \vw }\,:\, \vv,\vw \in \KK\;   \right\}\;\subset\; \CC^{2l \times 2}\;.
\end{equation} 
Note that
\begin{equation}\label{eq-form-T_0}
T^z_0 = \pmat{B_0 -zI & -I \\ I & \nul} \qtx{and}
\TT^z_0 = T^z_0\,\left\{ \pmat{\vx+\vv & \vw \\ \nul & \vx }\,:\, \vv,\vw \in \KK \right\} 
\end{equation}
where we adopt the notation that $T\Ab=\{TA\,:\, A \in \Ab\}$ for sets of matrices $\Ab$.

Moreover we consider the spectral measure $\mu_{\vx}$ at the vector $\vx \equiv \delta_0 \otimes \vx$, that means
$$
\int f d\mu_{\vx}\,=\, \langle \delta_0 \otimes \vx, \, f(H)\,(\delta_0 \otimes \vx)\, \rangle\;.
$$
Now,  using that the operator $H$ can not have compactly (finitely) supported eigenfunctions, Theorem~1 in \cite{Sadel2019} implies the following:
\begin{proposition}\label{th-spec-av-strip}  {\rm \cite{Sadel2019}}
In the sense of a weak limit for finite measures one finds that
\begin{align*}
{\rm d}\mu_{\vx}(\lambda)\,&=\,\lim_{n\to \infty} \,\frac{1}{\pi}\, \frac{{\rm d} \lambda}{ \min\limits_{T^\lambda \in \TT^\lambda_0}\| T^\lambda_{1,n} \, T^\lambda \smat{1 \\ 0} \|^2}\,=\,
\lim_{n\to \infty} \frac1 \pi\, \frac{{\rm d}\lambda}{\min\limits_{\vec v \in \KK} \left\|T^\lambda_{0,n}  \smat{\vx+\vec v \\ 0} \right\|^2 }
\end{align*}
\end{proposition}

Using the symplectic structure of the transfer matrices and the Banach-Alaoglu theorem  one can obtain a criterion for absolute continuity (see \cite{Sadel2019}).

\begin{proposition}\label{th-ac-criterion}
If one finds $\vec{u}_{\lambda,n} \in \CC^m$ for $\lambda \in (a,b)$, $n \in \NN$, such that
$$
\liminf_{n \to \infty} \;\int_a^b\, \left\|\,T^\lambda_{0,n} \smat{\vec{u}_{\lambda,n} \\ \vx}\,\right\|^4\;{\rm d}\lambda\,<\,\infty
$$
then, the measure $\mu_{\vx}$ is absolutely continuous in the interval $(a,b)$.
\end{proposition}

\begin{proof}
First, in \cite{Sadel2019} it was shown that the minimum $\min\limits_{\vec v \in \KK} \left\|T^\lambda_{0,n}  \smat{\vx+\vec v \\ 0} \right\|$ is achieved at a very specific vector which we call $\vec{v}_{\lambda,n}\in \KK$.
Defining
$$
f_n(\lambda):= \pi^{-1} \left\|T^\lambda_{0,n}  \smat{\vx+\vec v_{\lambda,n} \\ 0} \right\|^{-2}
$$
we see from Theorem~\ref{th-spec-av-strip} that $\mu_{\vx}$ is the weak limit of $f_n(\lambda) {\rm d}\lambda$ in the interval $(a,b)$.
Note that
\begin{align*}
&\left( T^\lambda_{0,n} \smat{\vec{u}_{\lambda,n} \\ \vx} \right)^* \smat{\nul & -I \\ I & \nul } T^\lambda_{0,n} \smat{\vx + \vec v_{\lambda,n} \\ 0} \,= \,\\
&\qquad =\, \smat{\vec{u}^*_{\lambda,n} & \vx\,^*}\, \big(T^\lambda_{0,n} \big)^*  \smat{\nul & -I \\ I & \nul} T^\lambda_{0,n}
\smat{\vx + \vec v_{\lambda,n} \\ 0} \,=\,\\
&\qquad =\, \smat{\vec{u}^*_{\lambda,n} & \vx\,^*} \smat{\nul & -I \\ I & \nul} 
\smat{\vx + \vec v_{\lambda,n} \\ 0} \, =\, \smat{\vec{u}^*_{\lambda,n} & \vx\,^*} \smat{0 \\ \vx + \vec v_{\lambda,n}}\,=\,1
\end{align*}
where we use $\|\vx\|=1$ and $\vx\,^* \vec v_{\lambda,n}=0$ as $\vec v_{\lambda,n} \in \KK$.
Now, using the Cauchy Schwartz inequality, this gives
$$
1\,\leq\, \left\|\,T^\lambda_{0,n} \smat{\vec{u}_{\lambda,n} \\ \vx}\,\right\|\; \cdot \;
\left\|T^\lambda_{0,n}  \smat{\vx+\vec v_{\lambda,n} \\ 0} \right\|
$$
and hence
$$
\pi^2 |f_n(\lambda)|^2\,=\,\frac{1}{\left\|T^\lambda_{0,n}  \smat{\vx+\vec v_{\lambda,n} \\ 0} \right\|^4}\,\leq\,
\left\|\,T^\lambda_{0,n} \smat{\vec{u}_{\lambda,n} \\ \vx}\,\right\|^4\;.
$$
Thus, the estimate given implies that
$$
\liminf_{n\to\infty} \int_a^b |f_n(\lambda)|^2\; {\rm d}\lambda\,<\, \infty\;.
$$
This means, along a suitable sub-sequence, the norm of $f_n$ in $L^2(a,b)$ is bounded.
By Banach-Alaaoglu, there is a sub-sequence ( o better, a sub-sub-sequence of the suitable sub-sequence) $f_{n_k}$ which converges weakly
in $L^2(a,b)$ to a limit $f \in L^2(a,b)$.
Noting that bounded continuous functions $g \in C_b(a,b)$ are also in $L^2(a,b)$ one has
$$
\lim_{k \to \infty} \int_a^b g(\lambda)\; f_{n_k}(\lambda)\,{\rm d}\lambda\,=\,
\int_a^b g(\lambda) \,f(\lambda)\,{\rm d}\lambda\;.
$$
for all $g\in C_b(a,b)$. But since $f_{n_k}(\lambda){\rm d}\lambda$ converges weakly to the measure $\mu_{\vx}$ this means that in the interval $(a,b)$ we have
$$
{\rm d}\mu_{\vx}(\lambda)\,=\,f(\lambda)\,{\rm d}\lambda\;
$$
which is an absolutely continuous measure in $(a,b)$ with a density in $L^2(a,b)$.
\end{proof}

\end{document}